\documentclass[5p,times]{elsarticle}
\usepackage{lineno,hyperref}
\modulolinenumbers[5]
\usepackage{comment}
\usepackage{booktabs, multirow} 
\usepackage{soul}
\usepackage[table]{xcolor} 
\usepackage{changepage,threeparttable}
\usepackage{amsmath}
\usepackage{amsthm}
\usepackage{graphicx}
\usepackage{mathrsfs}
\usepackage{graphicx}
\usepackage{rotating}
\usepackage{stackrel}
\usepackage{cancel}
\usepackage{multicol}

\usepackage{amsmath, amssymb, bbm, xspace}
\usepackage{epsfig}
\usepackage{longtable}
\usepackage{color}
\usepackage{mathrsfs}
\usepackage{comment}
\usepackage{ifthen}
\newboolean{showcomments}
\setboolean{showcomments}{true}
\usepackage{courier}
\usepackage{xcolor}
\usepackage{todonotes}
\usepackage[framemethod=tikz]{mdframed}
\usepackage{lineno}
\newmdenv[leftmargin=\dimexpr-0.4em, innerleftmargin=0.5em,
rightmargin=\dimexpr-0.4em, innerrightmargin=0.5em,
linewidth=2pt,linecolor=red, topline=false, bottomline=false,
innertopmargin=0pt,innerbottommargin=0pt,skipbelow=0pt,skipabove=0pt,%
]{notex}

\newenvironment{note}%
{\vskip\dimexpr\dp\strutbox-\prevdepth\relax\notex\strut\ignorespaces}%
{\xdef\notetpd{\the\prevdepth}\endnotex\vskip-\notetpd\relax}

\let\oldtodo\todo

\makeatletter%
\DeclareDocumentCommand{\todo}{ O{} +g +d<> }{%
		\setlength{\marginparwidth}{1.5cm}%
	\IfNoValueTF{#2}{\relax}{%
		\oldtodo[caption={#2},size=\scriptsize,#1]{\renewcommand{\baselinestretch}{0.8}\selectfont\sffamily#2\par}%
	}%
	\IfNoValueTF{#3}{\relax}{%
		\IfNoValueTF{#2}{
			\begin{note}%
				\begin{internallinenumbers}%
					\indent%
					#3%
				\end{internallinenumbers}%
			\end{note}%
		}{
			\vspace{-0\baselineskip}%
			\begin{note}%
				\begin{internallinenumbers}%
					\indent%
					#3%
				\end{internallinenumbers}%
			\end{note}%
		}%
	}%
}%
\makeatother
\usepackage{soul}

\newcommand{\hlc}[2][yellow]{{%
		\colorlet{foo}{#1}%
		\sethlcolor{foo}\hl{#2}}%
}

\newtheorem{conjecture}{Conjecture}[section]

\newcommand{\removetodo}[2]{\todo[color=pink]{\textbf{delete:} ``#1'' #2}\hlc[pink]{#1}}
\newcommand{\inserttodo}[1]{\todo[color=green!40]{\textbf{insert:} #1}}

\newcommand{\hltodoy}[2]{\todo[color=yellow!40]{#2}\hl{#1} }

\newcommand{\hltodoc}[3]{\todo[color=#3!40]{#2}\hlc[#3]{#1} }

\newcommand{\hltodo}[2]{\todo[color=orange!40]{#2}\hlc[orange!40]{#1} }
\newcommand{\replacetodo}[2]{\todo[color=pink!40]{\textbf{replace with:}``#2'' }\hl{#1} }

\usepackage{marginnote}
\newcommand{\todol}[1]{{%
		\let\marginpar\marginnote
		\reversemarginpar
		\renewcommand{\baselinestretch}{0.8}%
		\todo{#1}}}

\newcommand{\inserttodol}[1]{{%
		\let\marginpar\marginnote
		\reversemarginpar
		\renewcommand{\baselinestretch}{0.8}%
		\inserttodo{#1}}}

\newcommand{\removetodol}[2]{{%
		\let\marginpar\marginnote
		\reversemarginpar
		\renewcommand{\baselinestretch}{0.8}%
		\removetodo{#1}{#2}}}

\newcommand{\hltodol}[2]{{%
		\let\marginpar\marginnote
		\reversemarginpar
		\renewcommand{\baselinestretch}{0.8}%
		\hltodo{#1}{#2}}}

\newcommand{\replacetodol}[2]{{%
		\let\marginpar\marginnote
		\reversemarginpar
		\renewcommand{\baselinestretch}{0.8}%
		\replacetodo{#1}{#2}}}

\newcommand{\hltodoyl}[2]{{%
		\let\marginpar\marginnote
		\reversemarginpar
		\renewcommand{\baselinestretch}{0.8}%
		\hltodoy{#1}{#2}}}

\newcommand{\hltodocl}[3]{{		\let\marginpar\marginnote
		\reversemarginpar
		\renewcommand{\baselinestretch}{0.8}%
		\hltodoc{#1}{#2}{#3}}}

\newtheorem{theorem}{Theorem}[section]

\newtheorem{lemma}[theorem]{Lemma}

\newtheorem{proposition}[theorem]{Proposition}

\newtheorem{definition}{Definition}[section]

\def\bkE{{\rm I\kern-.17em E}}
\def\bk1{{\rm 1\kern-.17em l}}
\def\bkD{{\rm I\kern-.17em D}}
\def\bkR{{\rm I\kern-.17em R}}
\def\bkP{{\rm I\kern-.17em P}}

\def\bkZ{{\bf{Z}}}

\def\bkE{{\rm I\kern-.17em E}}
\def\bk1{{\rm 1\kern-.17em l}}
\def\bkD{{\rm I\kern-.17em D}}
\def\bkR{{\rm I\kern-.17em R}}
\def\bkP{{\rm I\kern-.17em P}}

\makeatletter
\newcommand{\pushright}[1]{\ifmeasuring@#1\else\omit\hfill$\displaystyle#1$\fi\ignorespaces}
\newcommand{\pushleft}[1]{\ifmeasuring@#1\else\omit$\displaystyle#1$\hfill\fi\ignorespaces}
\makeatother

\def\bkZ{{\bf{Z}}}
\def\b12{(\beta_1,\beta_2)}

\newenvironment{example}{{\noindent \bf Example}}{\hfill $\square$\hspace{-4.5pt}\vspace{6pt}}
\newcounter{example}
\renewcommand{\theexample}{\thesection.\arabic{example}}
\newenvironment{examplec}[1][]{\refstepcounter{example}
\par\medskip \noindent%
   \textbf{Example~\theexample. #1} \rmfamily}{\hfill $\square$   \hspace{-4.5pt} \vspace{6pt}}

\newcounter{remark}
\renewcommand{\theremark}{\thesection.\arabic{remark}}

\def\Bscr{\mathscr{B}}
\def\Xscr{\mathcal{X}}
\def\Yscr{\mathcal{Y}}
\def\indep{\perp\!\!\!\!\perp}

\newlength{\noteWidth}
\long\def\notes#1{\ifinner
{\tiny #1}
\else
\marginpar{\parbox[t]{\noteWidth}{\raggedright\tiny #1}}
\fi\typeout{#1}}

 \def\notes#1{\typeout{read notes: #1}} 

\newcommand{\I}[1]{\mathbb{I}_{\{#1\}}}

\newcommand{\ie}{i.e.\@\xspace} 
\newcommand{\eg}{e.g.\@\xspace} 

\def\spose#1{\hbox to 0pt{#1\hss}}

\def\text #1{\hbox{\quad#1\quad}}

\def\nthinsp{\mskip -2   mu}

\def\superstar{^{\raise 0.5pt\hbox{$\nthinsp *$}}}
\def\SUPERSTAR{^{\raise 0.5pt\hbox{$*$}}}

\def\lamstarT {\lambda^{\raise 0.5pt\hbox{$\nthinsp *$}T}}

\def\Ascr{{\cal A}}
\def\Bscr{{\cal B}}

\def\Pscr{{\cal P}}

\def\Wscr{{\cal W}}

\def\Xscr{{\cal X}}
\def\Yscr{{\cal Y}}

\def\eef{\;\textrm{if}\;}

\def\non{\nonumber}

\let\forallnew\forall
\renewcommand{\forall}{\forallnew\ }
\let\forall\forallnew

		\def\bkE{{\rm I\kern-.17em E}}
		\def\bk1{{\rm 1\kern-.17em l}}
		\def\bkD{{\rm I\kern-.17em D}}
		\def\bkR{{\rm I\kern-.17em R}}
		\def\bkP{{\rm I\kern-.17em P}}
		\def\bkY{{\bf \kern-.17em Y}}
		\def\bkZ{{\bf \kern-.17em Z}}
		\def\bkC{{\bf  \kern-.17em C}}

{\begin{list}{}%
         {\setlength{\leftmargin}{#1}}%
         \item[]%
}
{\end{list}}

		\def\bsp{\begin{split}}
		\def\beq{\begin{eqnarray}}
		\def\bal{\begin{align*}}
		\def\bc{\begin{center}}
		\def\be{\begin{enumerate}}
		\def\bi{\begin{itemize}}
		\def\bs{\begin{small}}
		\def\bS{\begin{slide}}
		\def\ec{\end{center}}
		\def\ee{\end{enumerate}}
		\def\ei{\end{itemize}}
		\def\es{\end{small}}
		\def\eS{\end{slide}}
		\def\eeq{\end{eqnarray}}
		\def\eal{\end{align*}}
		\def\esp{\end{split}}
		\def\qed{ \vrule height7.5pt width7.5pt depth0pt}  

	\def\cp2problem#1#2#3#4{\fbox
		 {\begin{tabular*}{0.9\textwidth}
			{@{}l@{\extracolsep{\fill}}l@{\extracolsep{6pt}}l@{\extracolsep{\fill}}c@{}}
				#1 & & $#4 $
			\end{tabular*}}}

		\def\bkE{{\rm I\kern-.17em E}}
		\def\bk1{{\rm 1\kern-.17em l}}
		\def\bkD{{\rm I\kern-.17em D}}
		\def\bkR{{\rm I\kern-.17em R}}
		\def\bkP{{\rm I\kern-.17em P}}

		\def\bkZ{{\bf{Z}}}

\newcommand {\beeq}[1]{\begin{equation}\label{#1}}
\newcommand {\eeeq}{\end{equation}}
\newcommand {\bea}{\begin{eqnarray}}
\newcommand {\eea}{\end{eqnarray}}

\def\texitem#1{\par\smallskip\noindent\hangindent 25pt
               \hbox to 25pt {\hss #1 ~}\ignorespaces}

\def\bsp{\begin{split}}
		\def\beq{\begin{eqnarray}}
		\def\bal{\begin{align*}}
		\def\bc{\begin{center}}
		\def\be{\begin{enumerate}}
		\def\bi{\begin{itemize}}
		\def\bs{\begin{small}}
		\def\bS{\begin{slide}}
		\def\ec{\end{center}}
		\def\ee{\end{enumerate}}
		\def\ei{\end{itemize}}
		\def\es{\end{small}}
		\def\eS{\end{slide}}
		\def\eeq{\end{eqnarray}}
		\def\eal{\end{align*}}
		\def\esp{\end{split}}
		\def\qed{ \vrule height7.5pt width7.5pt depth0pt}  

\usepackage{amsmath, amssymb, xspace}
\usepackage{epsfig}
\usepackage{longtable}
\usepackage{color}
\usepackage{mathrsfs}
\usepackage{subfig}

\journal{System and Control letters}

\begin{document}
	\begin{frontmatter}

		\title{{\bf Revisiting Common Randomness, No-signaling and  Information Structure in Decentralized Control}} 

\author{Apurva Dhingra\fnref{cminds}}
\address[cminds]{Center of Machine Intelligence and Data Science, CMInDS, Indian Institute of Technology Bombay, Mumbai 400076}
\ead{apurva.dhingra@iitb.ac.in}

\author{Ankur A. Kulkarni\fnref{cminds,sys,qui} \corref{id}}
\cortext[id]{Corresponding author}
\ead{kulkarni.ankur@iitb.ac.in}
\address[sys]{Systems and Control Engineering, Indian Institute of Technology Bombay, Mumbai 400076}
\address[qui]{Center of excellence in Quantum Information, Computing Science, and Technology, QuICST, Indian Institute of Technology Bombay, Mumbai 400076}

\begin{abstract}
	This work revisits the no-signaling condition for decentralized information structures. We produce examples to show that within the no-signaling polytope exist strategies that cannot be achieved by passive common randomness but instead require agents to either share their observations with a mediator or communicate directly with each other. This poses a question mark on whether the no-signaling condition truly captures the decentralized information structure in the strictest sense.


\end{abstract}


\end{frontmatter}



\section{Introduction}

We consider a decentralized team of agents bound by a static information structure. These agents have access only to their own information, without in-game access to a central coordinating device. The static information structure demands that agents cannot affect the observations of other agents, implying that agents cannot communicate with each other either.  Such scenarios arise, for example, in sensor networks like fire alarms.

In such a setting, the agents must devise a strategy  --  a conditional probability of choosing an action profile given an observation profile -- for optimizing a certain cost. Three main types of strategies exist: deterministic, behavioral, and local. Under a deterministic or pure strategy, agents choose their actions with certainty given their observations. In a behavioral strategy agents locally and independently randomize over actions, conditioned on their observations. Lastly, in a local strategy  agents have access to an external passive common randomness along with their own observation, and using these the agents choose their actions by locally and independently randomizing. All these strategies respect the information structure, since there is no in-game communication.

Another way of thinking of the information structure is as follows: the essence of the static information structure is that the information of the agent must be a sufficient statistic for its action. Consequently, any strategy must respect that an agent's action given its information must be independent of the other agents' information.
In physics, this mathematical condition is called the 'no-signaling condition', and is a universal law requiring that no information can travel instantaneously through space. There seems to be a belief in the control and physics communities that this no-signaling condition is equivalent to the requirement of the information structure, i.e., that the agents do not communicate during the game.

A seminal work  of Anantharam and Borkar in decentralized control \cite{ananthram2007commonrandom} takes this belief forward. It asks if every no-signaling strategy can be generated using passive common randomness. They show that there exist strategies within the no-signaling polytope that cannot be created with any amount of  common randomness. In other words, passive common randomness has limitations in accessing all distributions in the no-signaling polytope. Their finding emphasizes that specific strategies within no-signaling polytope are infeasible using mechanisms like behavioral and local strategies. This opens the question of whether there are more general mechanisms that respect no-signaling and information structure requirement; indeed there are quantum strategies, \cite{deshpande2022quantum,deshpande2022binary1,deshpande2022binary2},that do meet this requirement.

Our contribution in this paper is a re-look at the no-signaling condition, and asking if it does indeed capture the information structure. It is quite clear that satisfaction of the no-signaling condition is necessary for the information structure to hold. Our focus is on the \textit{converse}.
We demonstrate that within the no-signaling polytope, strategies exist that do not respect the information structure. In particular there are no-signaling strategies that can \textit{only} be created  using \textit{active} common randomness, \ie, a coordinating device that is \textit{dependent} on the observations of both agents. Such coordination allows for indirect exchange of information between agents, and is a violation of the information structure. We show also that the above no-signaling strategies that cannot be created using passive common randomness, \textit{can} be created if the passive common randomness is combined with one-way communication between the agents. In other words, if the information structure allows agents to play such strategies, it is effectively allowing them access to either an active coordinating device, or one-way communication. In summary, these results show that all strategies  within the no-signaling polytope are not necessarily feasible as per the `no communication' requirement of the information structure. A strategy that satisfies no-signaling conditions can still violate the static information structure.

We release this work with a tinge of diffidence: our results are shown by numerical construction of counterexamples and protocols, and the mathematics involved is not at all deep. However, to the best of our understanding, our results expose that the control community must either define the information structure requirements more broadly to match up with the no-signaling condition, or accept that the no-signaling conditions capture much more than the requirement of `no communication' in the strictest sense. They also highlight the need for further examination and a deeper understanding of the no-signaling condition.

This article has been organized as follows. Section 2 provides the background on the decentralized control setting. Section 3 includes the results of Anantharam and Borkar \cite{ananthram2007commonrandom}. It also features a schematic representation of the geometry of no-signaling. Section 4 presents the our main contributions, while Section 5 presents a `posterior' perspective on no-signaling. We conclude in Section 6 with a short discussion.

\section{Background}
\subsection{Notations}
We use $ (\Omega , \mathcal{F}, P)$ to denote a probability space, with $\Omega$ as the sample space, $\mathcal{F}$ as event space and $P$ as the probability measure. All random variables in this paper are discrete and all nonempty events take positive probability. We use capital letters $X,Y,$ and so on, to denote random variables, small letters $x,y$ to denote their realizations and the calligraphic letters $\mathcal{X}, \mathcal{Y}$ to denote the space where they take their values. We let $\Pscr(\Xscr)$ denote the set of probability distributions on $\Xscr$.
For $P \in \Pscr(\Xscr)$ we use $P(x)$ to denote the probability $P(X =x)$. When we want to make a statement for all realizations $x$ of $X$, we use $P(X)$.  With slight abuse of notation, we use $P(X)$ to also denote the probability distribution of $X$.
Similarly, $P(X |  A) $ denotes conditional probability of $X$ given $A$, and $P(x |  a) $ is a particular realization with $X=x$ and $A=a$. The notation `$X \indep Y\ |A$' stands for `$X$ and $Y$ are independent given $A$'.
We use $I(W; A, B)$ to denote the mutual information between $W$ and $(A, B)$. Similarly, $I(X; B |  A)$ denotes the conditional mutual information between $X$ and $B$ given $A$ (more details on properties of mutual information in \cite{cover2012elements}). For $a,b \in \{0,1\}$, $a \oplus b$ denotes addition modulo 2, and $a. \ b$ denotes logical \textit{AND} (multiplication) operation.

\newcommand{\A}{Venkat}
\newcommand{\B}{Vivek}
\newcommand{\Wit}{Hans}

\subsection{Common randomness and decentralized control}

Consider a decentralized system consisting of two agents, \A{} and \B{}. Each agent observes a random variable, denoted as $A \in \mathcal{A}$ and $B \in \mathcal{B}$ respectively, distributed according to a  joint distribution $P(A, B)$. The agents are required to choose their actions, denoted as $X \in \mathcal{X}$ and $Y \in \mathcal{Y}$ based on their private observations. A strategy for \A{} and \B{} is defined as a conditional joint distribution $P(X,Y |  A,B)$ of their actions $X,Y$ given their observations $A,B$.

These strategies satisfy the positivity and normalization conditions of any conditional probability distribution, and must also respect the information structure of the problem.  The decentralized information structure demands that
\A{} and \B{} do not communicate with each other after receiving their observations $A,B$.
For some strategies, \A{} and \B{} may be helped by common randomness $W$ drawn from some finite set $\Wscr$ provided by an agent \Wit{}\footnote{Through the choice of agent names, we wish to pay tribute to three of our scientific heroes: Venkat \textbf{A}nantharam, Vivek \textbf{B}orkar and Hans \textbf{W}itsenhausen.}. Notice that we have not defined what it means to ``communicate'', and indeed as we shall see in this paper, the root of our main observations lies precisely in defining what amounts to communication. For now, we shall stick to the colloquial meaning of communication.

This leads to various strategic classes for \A{} and \B{} defined by specific operational mechanisms that they can employ.

\def\Brs{\mathscr{B}}
\def\Lrs{\mathscr{L}}

\subsubsection{Behavioural strategies $\Brs$:}
In a behavioural strategy, \A{} and \B{} choose $X,Y$ upon observing $A,B$ by locally and independently randomizing.
In other words, $P$ satisfies,
\begin{align}\label{eq: local_stat}
P(X , Y |  A, B)
& = P(X |  A)   P(Y |  B).
\end{align}
The set of all such strategies is denoted by $\Brs$. One can see that  \eqref{eq: local_stat} is equivalent to the requirements
\begin{align*}
X & \indep Y \ |A,B, \\
X & \indep B \ |A, \\
Y & \indep A \ |B.
\end{align*}
The last two requirements above comprise the \textit{no-signaling} condition, a subject of our focus in this paper. We discuss them further in Section~\ref{sec: no-signal}.


\subsubsection{Local strategies $\Lrs$:} \label{sec:local_stat}
A local strategy $P(X , Y |  A, B) $  is  any strategy for which there exists $W \in \Wscr$, independent of $(A,B)$ with probability distribution $P(W)$, such that
\begin{align}\label{eq:local_poly}
P(X,Y, A,B,W) &= \ P(W)P(A,B) P(X |  A,W) \ P(Y |  B,W).
\end{align}
The above factorization is equivalent to,
\begin{align}
X &\indep Y  \  | A,B,W, \non \\
X &\indep B \  | A,W, \label{eq:nseq}\\
Y &\indep A \  | B,W, \non\\
W &\indep (A,B).
\end{align}
We thus get the strategy,
\begin{equation}
P(X,Y|A,B)= \sum_{W\in \Wscr} P(W)P(X|A,W)P(Y|B,W). \label{eq:local}
\end{equation}
The mechanism that creates a local strategy is as follows. Before heading to the field to observe $A,B$, \A{} and \B{} are given access to $W$, drawn  according to $P(W)$. When in the field, they make use of their knowledge of $W$ and their observations to choose their actions by locally and independently randomizing (as in strategies in $\Brs$). \Wit{} produces $W$ without reference to $(A,B)$; hence $W$ is independent of $(A,B)$. The set  of all local strategies forms a polytope, called the \textit{local polytope}, denoted as $\Lrs$. In fact $\Lrs$ is the convex hull of $\Brs$ (see \cite{saldi2022geometry}, \cite{deshpande2022quantum} for more details).

The random variable $W$ is called \textit{common randomness}. If, as above, $W$ is independent of $(A,B)$, then $W$ is called \textit{passive}, otherwise we refer to $W$ as \textit{active}. For passive common randomness $W$,
\begin{align}\label{eq: passive}
P(W,A,B) &= P(W) P(A,B),
\end{align}
while for active common randomness, this equation does not hold for some $a\in \Ascr,b\in \Bscr, w\in \Wscr$.

\def\NS{\mathscr{NS}}

\subsection{No-signaling strategies $\NS$}\label{sec: no-signal}
We now come to the question of information structure, or equivalently communication, and its meaning. Concretely, we ask: what is the set of \textit{all} strategies that are allowed by the information structure of our problem? In stochastic control, it is common to define the allowed strategy space by means of explicit mechanisms, as in $\Brs$ and $\Lrs$. One can of course define specific mechanisms that respect the information structure, but what is to say that more general mechanisms cannot be devised? Indeed more general mechanisms \textit{are} known. A case in point are those that deploy entangled \textit{quantum} systems to create randomness, see \eg, \cite{deshpande2022quantum}. These too respect the information structure and describe a set strictly larger than $\Lrs$. While that does push the boundary, it still fails to answer what is the universal set of strategies that \A{} and \B{} can play?

This brings us to one of the central topics of this paper, the no-signaling condition.
\begin{definition}
	A strategy distribution $P(X , Y |  A, B)$ is said to be no-signaling if  \A{}'s action $X$  is independent of $B$ given $A$, and \B{}'s action $Y$ is independent of $A$ given $B$, \ie,
	\begin{align}
		P(X |  A,B) = P(X |  A),  \nonumber\\
		P(Y |  A,B) = P(Y |  B). \label{eq:ns}
	\end{align}
	These  are called no-signaling conditions and we denote by $\NS$ the set of all strategies that satisfy \eqref{eq:ns}.
\end{definition}
When players cannot communicate with each other, each player's observation must be a sufficient statistic for its action. This is precisely the condition demanded by \eqref{eq:ns}. It appears plausible that $\NS$ is in fact the set of strategies that respect the information structure. Indeed in \cite{ananthram2007commonrandom}\footnote{p.\ 570, first column at the bottom of the page}, Anantharam and Borkar say,
\begin{quote}
	The salient characteristic of the distributed creation of the pair $(X, Y )$ from
	$(A, B)$ is that $X$ is created with access to $A$ but without reference to $B$ and $Y$ is created with access to $B$ but without
	reference to $A$. Thus it is natural to conjecture that for every
	$(X, Y, A, B)$ satisfying the conditions,
	\begin{align}
		(A,B) & \sim P(A,B),\non \\
		I(X;B|A) & = 0,  \label{eq:ananthcommon} \\
		I(Y;A|B) & = 0, \non
	\end{align}
	there must exist a $W$ (on a possibly augmented
	sample space) such that $(X, Y, A, B, W )$ satisfy conditions \eqref{eq:local_poly}\footnote{From the present paper}.
\end{quote}
Note that \eqref{eq:ananthcommon} is equivalent to \eqref{eq:ns}.
Anantharam and Borkar proceed to give a counterexample to the above conjecture. We discuss this counterexample in the next section. For now we note that Anantharam and Borkar are implicitly implying that \textit{every} strategy in $\NS$ is allowed by the information structure; indeed it is worthwhile seeking a $W$ as above only if one believes that every strategy in $\NS$ is a valid strategy for the problem. Any strategy that respects the information structure of the problem is indeed no-signaling (it is easy to verify that $\Brs \subseteq \Lrs \subseteq \NS $; quantum strategies from~\cite{deshpande2022quantum} are also no-signaling). However, Anantharam and Borkar seem to have in mind that the \textit{converse} also holds.

The term `no-signaling' originated in physics as a universal law that cannot be violated by any physical system. In physics too there seems to be a similar belief. In \cite[p.\ 80,  first paragrpah of Section 4.4]{laloe2012we}, Laloe says,
\begin{quote}
	The theory of relativity requires that it should be fundamentally impossible to transmit signals containing information between two distant points faster than the speed of light (relativistic causality); suppressing this absolute limit would lead to serious internal inconsistencies in the theory.

	$\ldots$

	From a general point of view, and without restricting the discussion to quantum
	mechanics, what are the general mathematical conditions ensuring that a theory is
	“non-signaling” (NS conditions), meaning that it does not allow Alice to transmit
	instantaneous signals to Bob (and conversely)?

	$\ldots$

	When the experiment is repeated, since Bob does not have access to Alice’s
	results, the only thing that he can measure is the occurrence frequency of his own
	results. This corresponds to the probabilities obtained by a summation over $X$ of
	the preceding probabilities (sum of probabilities associated with exclusive events):
	\begin{equation*}
		\sum_X P(X,Y|  A,B)
	\end{equation*}
	The NS condition amounts to assuming that this probability is independent of $a$;
	we therefore obtain the condition:
	\begin{equation*}
		\sum_X P(X,Y|  a,b) = \sum_X P(X,Y|  a',b) \ \ \forall  \ b.
	\end{equation*}
\end{quote}
Note that the above equation is in fact Eq \eqref{eq:ns}.

From the above quotes it is apparent that the view in the scientific community at large is that $\NS$ is the set of strategies allowed by the information structure. Our main finding is that things are not as simple. There are strategies that allow \A{} and \B{} to exchange information that are also included in $\NS$. These include strategies where \textit{active} common randomness is provided by \Wit{}, and it is exploited by \A{} and \B{} in choosing their actions. Another possible protocol is using passive common randomness and classical one-way communication from \A{} to \B{}.

\section{No-signaling with passive common randomness}\label{sec: method1}
\def\strat{P(X,Y |  A,B)}

\subsection{Anantharam and Borkar's conjecture and counterexample}
As mentioned in the previous section, Anantharam and Borkar seem to take the view that every strategy in $\NS$ is allowable in the information structure. Taking this thought forward, they next ask what mechanism may the agents use to achieve strategies in $\NS$. They make (the very reasonable) conjecture that every point in $\NS$ is achievable through passive common randomness.

\begin{conjecture}\label{conj:conj0}\cite{ananthram2007commonrandom}
	Given any strategy $P(X , Y |  A, B) $  that satisfies the  no-signaling condition \eqref{eq:ns}, there exists some passive common randomness $W$ and a joint distribution on $(X,Y,A,B,W)$ such that $\strat$ satisfies \eqref{eq:local}.
\end{conjecture}

Anantharam and Borkar present a counterexample to this conjecture. We recount it below.

\begin{examplec}\label{Table : ccex} This is a counterexample to Conjecture \ref{conj:conj0}.
	\begin{table}[!htp]
		\centering
		\scriptsize
		\begin{center}
			\begin{tabular}{c c c c c c  c c c c c}
				\midrule
				$(a,b)$        &  &               & $P(x,y |  a,b)$   &                & \quad & $(a,b)$ &  &               &  $P(x,y |   a,b)$   &                \\
				[1ex]
				\midrule
				&  & $\frac{1}{3}$ & 0             & 0              &       &       &  & 0             & $\frac{1}{3}$ & 0              \\[1ex]
				(1,1)        &  & 0             & $\frac{1}{3}$ & 0              &       & (1,0) &  & $\frac{1}{3}$ & 0             & 0              \\[1ex]
				&  & 0             & 0             & $\frac{1}{3} $ &       &       &  & 0             & 0             & $\frac{1}{3} $ \\
				[1ex]
				\midrule
				&  & 0             & $\frac{1}{3}$ & 0              &       &       &  & 0             & 0             & $\frac{1}{3}$  \\
				[1ex]

				(0,1) &  & 0             & 0             & $\frac{1}{3}$  &       & (0,0) &  & 0             & $\frac{1}{3}$ & 0              \\[1ex]
				&  & $\frac{1}{3}$ & 0             & 0              &       &       &  & $\frac{1}{3}$ & 0             & 0              \\[1ex]
				\midrule

			\end{tabular}
		\end{center}
		{\raggedright
			In this table, the rows of $P(X,Y |  A,B)$ are for index $x \in \Xscr = \{0,1,2\}$ and columns for $y \in \Yscr = \{0,1,2\}$. For this example $P(a,b) = \frac{1}{4} \ \forall a,b$. \par}
	\end{table}
	The distribution $P(X, Y |  A,B)$ shown in Table \ref{Table : ccex} follows the no-signaling conditions \eqref{eq:ns}. Specifically, $P(x |  a,b) =\frac{1}{3} \ \forall \ x, a,b$ whereby $P(X|A,B)\equiv P(X |  A) $ and similarly  $P(Y |  A,B)  = P(Y |  B) $. However, it is impossible to find a common randomness $W$ such that the joint distribution factorizes as \eqref{eq:local_poly}. Details of this can be found in \cite{ananthram2007commonrandom}.
\end{examplec}

This counterexample implies that there are strategies within $\NS$ that can not be accessed using common randomness.
Anantharam and Borkar view this result as a limitation of common randomness.

Below we present another counterexample to Conjecture \ref{conj:conj0}, showing that the conjecture is false even for binary spaces.

\begin{examplec} \label{Table : ccex2}
	This is another counterexample to Conjecture
	\ref{conj:conj0} with strategy distribution  $P(X,Y |  A,B)$ shown as in \ref{Table : ccex2}.

	\begin{table}[!htp]
		\centering
		\begin{center}
			\begin{tabular}{c c c c c  c c c c}
				\midrule
				           $(a,b)$             &         & $P(x,y |   a,b)$ &               & \quad         & $(a,b)$ &               & $P(x,y |   a,b)$ &               \\
				[1ex]
				\midrule
				                                        &                  & $\frac{1}{2}$ & $0$   &        &         & $\frac{1}{2}$ & $0$              &  \\
				            [1ex]               $(0,0)$ &                  & $\frac{1}{3}$ & $\frac{1}{6}$ && $(0,1)$               & $0$              & $\frac{1}{2}$ \\
				      [1ex]
				\midrule       &               & $\frac{1}{2}$    & $0$  &         & & $0$              & $\frac{1}{2}$                \\
				      [1ex]
				$(1,0)$ &               & $\frac{1}{3}$    & $\frac{1}{6}$ &&$(1,1)$        &         $\frac{1}{2}$    & $0$          \\
				            [1ex]
				            \hline
			\end{tabular}
		\end{center}
		\scriptsize
		{\raggedright Here, $x \in \{0,1\}$ is represented row-wise and  $y \in \{0,1\}$  are indexed column-wise in the strategy distribution $P(X,Y |  A,B)$, where $a,b \in \{0,1\}$. Note that the no-signaling condition is satisfied $ P(x |  a,b)  \ = \  \frac{1}{2} \ \forall \ x,a,b \  $ and $ \  P(y |  b,a)= p(y|b)$. 
			\par}
	\end{table}
This example satisfies the no-signaling condition \eqref{eq:ns}, as shown below,
\begin{align}
    P(x|a,b) &= P(x|a)=  \frac{1}{2} \forall a,b,x \label{eq:ex-ns}\\
    P(y|a,b) & = \
    \begin{cases}
        \frac{1}{2}  \  \eef b =1 \ \forall \ y,a\\
        \frac{1}{6}  \ \eef y=1 , b=0 \ \forall \  a \\
        \frac{5}{6}  \ \eef y = 0 , b=0  \ \forall \  a
    \end{cases}\label{eq:ex-ns2}\\
    &= P(y|b) \nonumber \ \forall \ y,b
\end{align}
A proof that this counterexample lies in $\NS \backslash \Lrs$ is shown in Section \ref{sec: appendix}.
\end{examplec}

In the following section we discuss the geometry of no-signaling distributions  \cite{popescu1994nonlocquantum} and non-local quantum correlations \cite{barrett2005nonlocal} as described below.

\subsubsection{Geometry of no-signaling and Counter-example \ref{Table : ccex}} \label{sec: geomtry_NS}
The geometry of the $\NS$ polytope becomes significantly more complicated with increasing size of the spaces $\Ascr,\Bscr,\Xscr,\Yscr$. We give a description of this geometry for the case where these sets are binary.
For such action and observation spaces the extreme points of $\NS$ are well understood and classified as \textit{local} and \textit{non-local} \cite{barrett2005nonlocal}, as described below.
\begin{itemize}
\item \textbf{Local vertices}: These refer to the extreme points of the local polytope, $\Lrs$, or deterministic strategies in $\Lrs$ . There are a total of 16 local extreme points expressed as,
\begin{align}
	P(X , Y |  A, B)  &= \begin{cases}
		&1  \ , \  X = \alpha \ A \ \oplus\  \beta \ , \ \  Y = \gamma \ B \ \oplus  \ \delta  \\
		&0  \ ,  \text{else}
	\end{cases}\label{eq:local_ver}
\end{align} where $\alpha, \beta , \gamma , \delta , \in \{0,1\}$ and $x \in \Xscr,y \in \Yscr,a \in \Ascr,b \in \Bscr$.

\item \textbf{Non-local vertices}: These refer to the extreme points of the no-signaling polytope, $\NS$, excluding the local vertices. There are total 8 non-local vertices that can be expressed in the form,
\begin{align}\label{eq:non-local}
	P(X , Y |  A, B) &=\begin{cases}
		&\frac{1}{2} \  , \  X\ \oplus Y  = A \ .\ B \oplus \alpha \ A \oplus \beta \ B \ \oplus  \ \gamma  \\
		&0  \ ,\text{else}
	\end{cases}
\end{align} where $\alpha, \beta , \gamma , \in \{0,1\}$ and $x \in \Xscr,y \in \Yscr,a \in \Ascr,b \in \Bscr$.
\end{itemize}
Note, all of the local vertices are also extreme points of the $\NS$, and $\mathscr{L} \subset \mathscr{NS}$. This geometry is shown in the schematic diagram in Figure~\ref{fig} as per \cite{barrett2005nonlocal},\cite{deshpande2022quantum}. 

\begin{figure}
\centering
\includegraphics[width=0.6\linewidth]{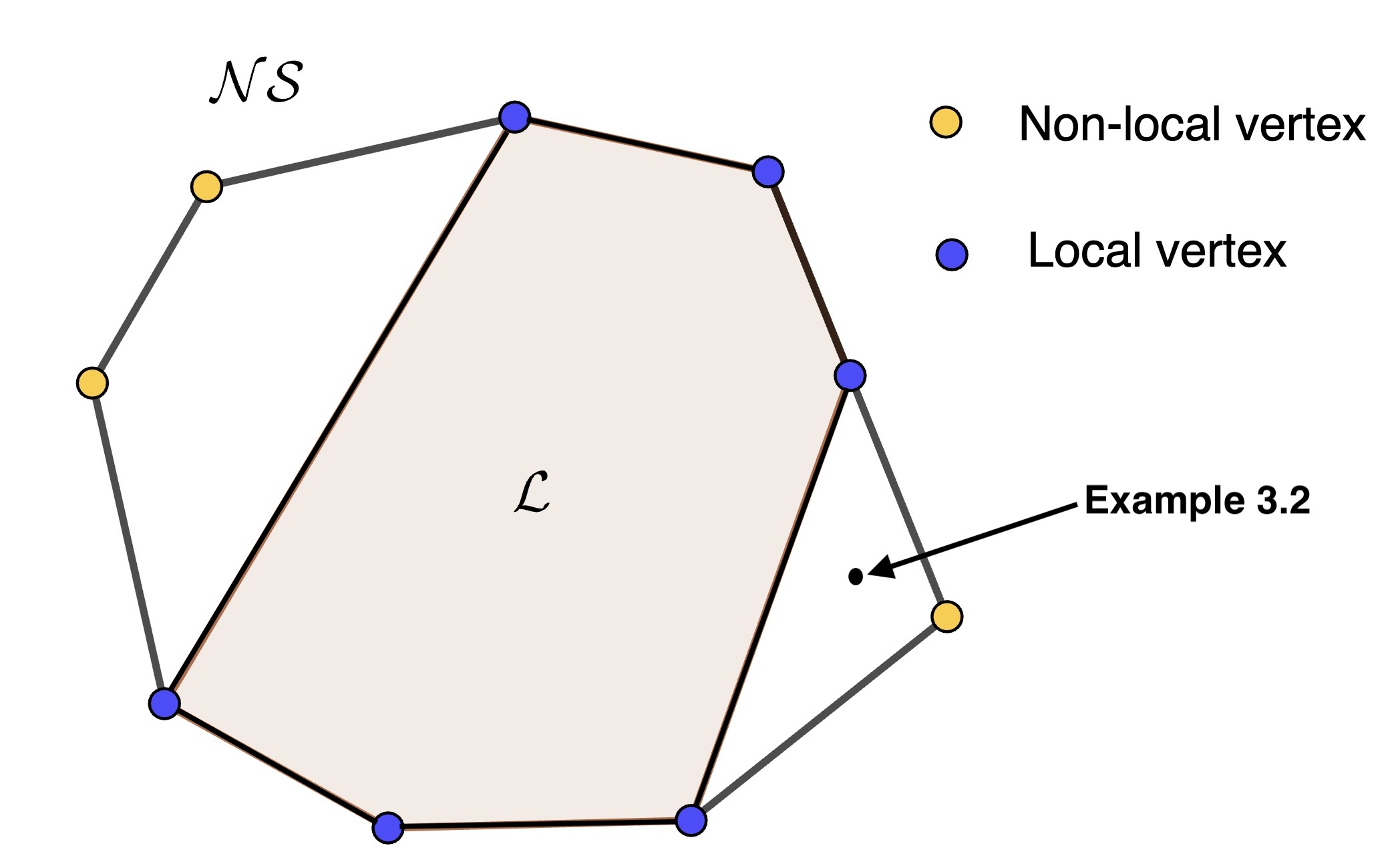}
\caption{A schematic representation of no-signaling polytope. L and NL label the vertices for local and non-local points. The set bounded by local vertices is $\Lrs$ and overall is the no-signaling polytope $\NS$. Example \ref{Table : ccex2} represents  a distribution $P(X,Y|A,B) \in \NS \backslash \Lrs$.}  
\label{fig}
\end{figure}

\section{No-signaling with active common randomness}
As evident so far, there are no-signaling strategies that lie outside $\Lrs$. This begs the following question: if \A{} and \B{} were allowed access to \textit{active} common randomness, what strategies would they achieve? We investigate this in the following section.

\subsection{No-signaling with active common randomness} \label{sec: method2}
Consider the following mechanism involving active common randomness.  \Wit{} is allowed to view the realized value of the pair $(A,B)$ and generates a $W$ whose distribution \textit{depends} on this value.  The realized value of $W$ is informed to \A{} and \B{}, but they are not informed of the realized value of $(A,B).$ \A{} and \B{} then set out into the field where they receive realized values of $A,B$ respectively. Thus in the field, \A{} has access to $A,W$ and \B{} has access to $B,W.$ They now generate $X,Y$ by independently and locally randomizing using this information, without any communication with each other. Under this  mechanism, their strategy $\strat$ takes the form
\begin{equation} \label{eq:pdf1}
\strat = \sum_W P(W | A, B)  P(X  |  A,W)  P(Y  |  B,W).
\end{equation}
This factorization satisfies the first three conditions from \eqref{eq:nseq}.
Indeed, it appears that access to a $W$ that depends on $(A,B)$ would have allowed \A{} and \B{} to gain access to some bits of each other's information and thereby the no-signaling condition would now be violated.
Or will it? We show below that there are strategies for which the no-signaling condition \textit{still holds}, even when players have access to active common randomness. In fact, we show something stronger: there are strategies in $\NS$ for which active common randomness is \textit{essential}, \ie, these strategies can be generated with access to active common randomness but cannot be generated with \textit{any} passive common randomness.

In other words, we show that the following conjecture is false.
\begin{conjecture}\label{conj:conj1}
If $P(X,Y |  A,B) $ satisfies \eqref{eq:pdf1} with active common randomness $W$ and $P(X,Y |  A,B) \notin \Lrs$, then $P(X,Y|  A,B)$ does not belong to $\NS$.
\end{conjecture}

\subsubsection{Counterexample of Conjecture \ref{conj:conj1}:}
The strategy $\strat$ shown in Example[\ref{Table : ccex2}] serves as our counterexample. Note that $\strat \in \NS \backslash \Lrs$, whereby it suffices to construct a $W$ such \eqref{eq:pdf1} holds for this strategy.

We consider active common randomness $W$ with $\mathcal{W} = \{1,2,3,4\}$. We take $P(w|a,b)$ as
\begin{align}	\label{eq:pw2}
P(w| a = 0 , b=1) &= \frac{1}{4}  = P(w| a=1, b=1) \ \forall \ w \nonumber \\
	P(w| a=1,b=0) &= \
	\begin{cases}
		& {\small \frac{1}{2}}  , \ w=1  \\
		&{\small \frac{1}{3}}  , \ w=2\\
        		&{\small 0}  , \ w=3\\
		&{\small \frac{1}{6}}   , \ w=4
	\end{cases}\nonumber\\
P(w| a = 0,b =0) &= \
\begin{cases}
	&{\small \frac{1}{3}}  , \ w=1\\
    	&{\small \frac{1}{2}}   , \ w=2  \\
	&{\small \frac{1}{6}}   , \ w=3 \\
	&{\small 0}  , \ w=4
\end{cases}\ 
\end{align}
whereby $W$ is clearly active.
As we want to satisfy \eqref{eq:nseq}, that can also be written as,
\[P(x,y | a,b,w) = P(x|a,b,w) P(y|a,b,w) =P(x|a,w) P(y|b,w), \]
 it is enough to define the marginal distributions  $P(x|  a,w) , P(y | b,w)$. For ease of exposition, we represent $P(x|a,w)$ as a column vector where rows are indexed as $x=\{0,1 \}.$
\begin{align}
	&P(x|a=0,w=1,3) = 	\begin{bmatrix} 0 \\ 1 \end{bmatrix}
	= P(x|a=1,w=2,4) \nonumber\\
	&P(x|a=1,w=1,3) =  	\begin{bmatrix}1 \\0	\end{bmatrix}
	= P(x|a=0,w=2,4) \label{eq:w1_x}
	\end{align}
    For each fixed $x,a,w$, we define $P(x|a,b,w) = P(x|a,w)\ \forall \ b.$
	Similarly, we define $P(y|a,b,w) = P(y|b,w) \ \forall \  a$. The marginal distribution $P(y|b,w)$ is shown below, represented as a row vector with each column relating to $y \in \{0,1\}$.
	\begin{align}
		&P(y|b=0,w=1,2) = 	\begin{bmatrix}	1 &	0 \end{bmatrix}
		= P(y|b=1,w=2,4)\nonumber\\
		&P(y|b=1,w=1,3) =  	\begin{bmatrix} 0 & 1	\end{bmatrix}
		= P(y|b=0,w=3,4) \label{eq:w1_y}
		\end{align}

We now check that \eqref{eq:pdf1} holds. For any fixed $a,b$ we express $P(x,y|a,b)$ as a matrix where rows represent $x\in \{0,1\}$ and columns represent $y\in \{0,1\}$. Note, the product of $P(y|  b ,w) P(x |  a,w)$ is treated as an outer product between a column and row vector, for the sake of representation. Consider $	a = 1 $ and $  b = 1 $, using \eqref{eq:pw2} - \eqref{eq:w1_y}, we have
\begin{align*}
\sum_{w \in \Wscr}  P(x &|  a=1,w) P(y |  b =1,w)P(w |  a=1,b=1) \\
&=  \frac{1}{4} \sum_{w \in {1,2,3,4}} P(x |  a=1,w) P(y|  b =1,w) ,\\
&= \frac{1}{4}  \left( \begin{bmatrix} 1 \\ 0 \end{bmatrix}  	\begin{bmatrix} 0 & 1\end{bmatrix} +
\begin{bmatrix} 0 \\ 1 \end{bmatrix}  	\begin{bmatrix} 1 & 0\end{bmatrix} +
\begin{bmatrix} 1 \\ 0 \end{bmatrix}  	\begin{bmatrix} 0 & 1\end{bmatrix}+
\begin{bmatrix} 0 \\ 1 \end{bmatrix}  	\begin{bmatrix} 1 & 0\end{bmatrix} \right)\\
&= \frac{1}{2}\begin{bmatrix}0  & 1 \\1 & 0\end{bmatrix} \\ &= P(x,y |  a = 1 , b = 1).
\end{align*}
The last equality can be inferred by inspection with the top-left matrix in the table in Example \ref{Table : ccex2}. Similarly, for $a=0$ and $b=1$,
\begin{align*}
	\footnotesize
	\sum_{w \in \Wscr}  P(x& |  a=0,w)P(y|  b =1,w)  P(w |  a=0,b=1) \\
	&=  \frac{1}{4} \sum_{w \in {1,2,3,4}}P(x |  a=0,w)P(y|  b =1,w) ,\\
	&= \frac{1}{4}  \left( \begin{bmatrix} 0 \\ 1 \end{bmatrix}  	\begin{bmatrix} 0 & 1\end{bmatrix} +
	\begin{bmatrix} 1 \\ 0 \end{bmatrix}  	\begin{bmatrix} 1 & 0\end{bmatrix} +
	\begin{bmatrix} 0 \\ 1 \end{bmatrix}  	\begin{bmatrix} 0 & 1\end{bmatrix}+
	\begin{bmatrix} 1\\ 0 \end{bmatrix}  	\begin{bmatrix} 1 & 0\end{bmatrix} \right)\\
	&= \frac{1}{2}\begin{bmatrix}
		1  & 0 \\
		0 & 1
	\end{bmatrix} = P(x,y |  a = 0 , b = 1),
\end{align*}
as required.
Now for $a = 0 $ and $  b = 0 $ we have,
\begin{align*}
\sum_{w \in \Wscr}P(x &|  a=0,w) P(y |b =0,w)  P(w |   a,b) \\
&= \sum_{i \in \{1,2,3\}}P(x |  a=0,w=i) P(y  |  b =0,w=i)  P(w=i|a,b)\\ 
&= \frac{1}{3} \begin{bmatrix} 0 \\ 1 \end{bmatrix}  	\begin{bmatrix} 1 & 0\end{bmatrix} +
\frac{1}{2}  \begin{bmatrix} 1 \\ 0 \end{bmatrix}  	\begin{bmatrix} 1 & 0\end{bmatrix} +
\frac{1}{6} \begin{bmatrix} 0 \\ 1 \end{bmatrix}  	\begin{bmatrix} 0 & 1\end{bmatrix} \\
&=\begin{bmatrix}
\small \frac{1}{2}  & 0 \\
\small \frac{1}{3} & \small \frac{1}{6}
\end{bmatrix} = P(x,y  | a = 0, b = 0).
\end{align*}
Similarly for $a=1$ and $ b = 0$, we have
\begin{align*}
\sum_{w \in \Wscr}P(x& |  a=0,w) P(y |b =1,w)  P(w | a,b) \\
&=   \sum_{i  \in \{1,2,4\}}P(x |  a=0,w=i) P(y  |  b =1,w=i) P(w =i| a,b) \\
&= \frac{1}{2} \begin{bmatrix} 1\\ 0 \end{bmatrix}  	\begin{bmatrix} 1 & 0\end{bmatrix} +
\frac{1}{3}  \begin{bmatrix} 0 \\ 1 \end{bmatrix}  	\begin{bmatrix} 1 & 0\end{bmatrix} +
\frac{1}{6} \begin{bmatrix}  0\\ 1 \end{bmatrix}  	\begin{bmatrix} 0 & 1\end{bmatrix} \\
&=\begin{bmatrix}
	\footnotesize
\small \frac{1}{2 } & 0 \\
\small \frac{1}{3} & \small \frac{1}{6}
\end{bmatrix} = P(x,y  | a = 0 , b = 1).
\end{align*}
Thus, \eqref{eq:pdf1} holds for the marginal distributions \eqref{eq:w1_x}, \eqref{eq:w1_y}, creating the Example \ref{Table : ccex2}. Hence, this counterexample shows that even after giving active common randomness to the agents the resulting strategy still satisfies the no-signaling condition. Moreover, this strategy can only be generated using active common randomness.

\subsubsection{Where is the communication?}
Let us now understand what form of information exchange or communication has taken place between the agents. To answer this question, one needs to be precise about the parties involved in the communication and the timeline according to which parties access information. Recall that we assumed \Wit{} is informed of $A,B$. How does \Wit{} get access to this information? There are two possibilities that we can think of:

\begin{enumerate}
\item[(a)] In the first possibility, $A,B$ need to be sensed in the field, and the only agents that can do this sensing are \A{} and \B{}. Communication causally follows the sensing and it is between \A{} and \Wit{}, and \B{} and \Wit{}, respectively. $W$ is then realized and is communicated back from \Wit{} to \A{} and \B{}, respectively. This is in the same spirit as confidential communication with a mediator in incomplete information games~\cite{myerson1997game}; in this case $P(W|A,B)$ can be thought of as a Bayesian correlated strategy.
\item[(b)] Another possibility is that \Wit{}  may know of $A,B$ through a separate channel. Here there is communication between the environment that generates $(A,B)$ and \Wit{}, and one-way communication from \Wit{} to \A{} and \B{}.
\A{} and \B{}  do not communicate with each other after they are informed of $W$ and their observations.
\end{enumerate}

In both cases \A{} and \B{} do not communicate with each other \textit{directly}. In case (a) there is two-way private communication with \Wit{}, whereas in case (b), \Wit{} has separate access to $A,B$. If case (b) is the view taken, then one must ask why it does not degenerate the decentralized nature of the problem. Indeed if \Wit{} could access both $A,B$, he may as well signal the optimal centralized action to the players via $W$.
On the other hand if case (a) is the view of the mechanism, we are compelled to ask if such communication is permitted as part of the definition of the information structure.
Literally read, the information structure only says that players ``cannot communicate \textit{with each other}'', leaving a loophole where players could communicate confidentially with a mediator. The no-signaling condition allows for strategies where players exploit this loophole. That $\NS$ also contains such strategies is our first main finding of this paper.

To us, this finding highlights the following:
\begin{enumerate}
\item In defining the information structure of the problem, one must have clarity on what amounts to communication and what exact communication is permitted by the information structure.
\item One must accept that the no-signaling condition is not tantamount to ``no communication'' in the wide sense of the word ``communication''.
\end{enumerate}

Our leaning is that confidential communication as in case (a) is also a violation of the information structure. For instance, consider that the set of $W$'s is as large as the set of all observations, \ie, $|\mathcal{W}| = |\Ascr \times \Bscr|$, then when \Wit{}  could encode $W$ to include all information about the pair $(A,B)$, whereby \A{} and \B{} would indirectly know the observation of the other via $W$, and when they choose their respective actions $X,Y$, using their own observation and $W$, they have effectively referenced the observation of other agent.  We feel this should not be permitted by the information structure.

A natural question arises: how does the no-signaling condition not get violated in spite of this information exchange? The answer to this lies in the nature of the condition itself. No-signaling is a \textit{statistical condition}, because it talks of the \textit{probability measure} of the random variables involved. In a frequentist view, this says that \textit{on average} $X$ must be independent of $B$ given $A$ (and similarly for $Y$). This does not preclude that there are sample-paths for which some form of communication has happened. Contrast this with specific mechanisms which are expressed in terms of random variables instead.

\subsection{Strategy within $\NS$ with one-way communication} \label{sec:one-way}
Consider the same strategy $P(X,Y |  A,B)$ as given in Example \ref{Table : ccex2}, considering $x,y \in \{0,1\}$ and $a,b \in \{0,1\}$. We now show a protocol where this distribution is created by \A{} and \B{} using one-way communication and passive common randomness. The process is as follows,
\begin{itemize}
\item \Wit{} produces passive common randomness $W \in \{1,2,3\}$, uniformly at random. He shares the realization of $W$ with \A{} and \B{}.
\item During the game, \A{} makes his observation $A$. He generates his action  $X$ as per the following equation, for $W=1,2$,
\begin{align}
	X   = \
	\begin{cases}
		A \oplus 1, &\eef W=1\\
		A, &\eef W=2
	\end{cases} \label{eq: x_one_way}
\end{align}
For $W=3$, the action $X$ is generated uniformly at random, i.e $X \in \{0,1\}$ is equally likely (and hence not dependent on $A$).

\item \A{} shares the realization of his observation $A$ and the action $X$ with \B{} when $W=3$, otherwise there is no communication or sharing.
\item \B{} makes his observation $B$. Based on his available information he generates his action $Y$ as per the following equation,
\begin{align}\label{eq: y_one_way}
	Y  & = \
	\begin{cases}
		B &\eef W=1 \\
		0 &\eef W=2 \\
		X \oplus A.B &  \eef   W = 3
	\end{cases}
\end{align}
Note: For $W=3$, \A{} has shared $X,A$ with \B{}, so this is the case of one-way communication.
\end{itemize}
We now show that this process creates the strategy $P(X,Y |  A,B)$  from Example \ref{Table : ccex2}. We have,
\begin{align}\label{eq:com}
	P(X,Y|  A,B) & = \sum_{W \in \{1,2,3\}} P(X,Y |  A,B,W) P(W) \nonumber  \\
	&= \frac{1}{3} \sum_{W \in \{1,2,3\}}  P(X,Y |  A,B,W).
\end{align}
Comparing \eqref{eq: x_one_way} and \eqref{eq: y_one_way} with marginal distribution \eqref{eq:w1_x}, \eqref{eq:w1_y} for $w = 1,2$, we note that $P(X,Y|A,B,W)$ for $W=1,2$ here takes the form as shown in Table \ref{Table:2w3}.

		\begin{table}[ht]
		\centering		
		\caption{Distribution $P(x,y|  a,b,w=1,2)$ with common randomness $W$  }
		\label{Table:2w3}
		\begin{tabular}{c|c|c|c}			
			\hline
			\multicolumn{2}{c|}{$P(x,y|a,b,w=1)$}                                                              &                      \multicolumn{2}{c}{$P(x,y | a,b,w=2)$}                                           \\[0.5ex] 
			\hline
			$a=0,b=0$ & $a=0,b=1$ & $a=0,b=0$ & $a=0,b=1$\\[0.5ex] \hline    
			& &&\\                                    
			$\begin{bmatrix}0  & 0 \\	1 & 0 \end{bmatrix}$& 
			$\begin{bmatrix}0  & 0 \\	0 & 1 \end{bmatrix}$ &                  
			 $\begin{bmatrix}1  & 0 \\	0 & 0 \end{bmatrix}$ & 
			 $\begin{bmatrix}1  & 0 \\	0 & 0 \end{bmatrix}$                \\
			& &&\\ \hline
			$a=1,b=0$ & $a=1,b=1$ & $a=1,b=0$ & $a=1,b=1$\\[0.5ex] \hline
			 & &&\\
			 $\begin{bmatrix}1  & 0 \\	0 & 0 \end{bmatrix}$& 
			 $\begin{bmatrix}0  & 1 \\	0 & 0 \end{bmatrix}$               &       
			 $\begin{bmatrix}0  & 0 \\	1 & 0 \end{bmatrix}$   &  
			 $\begin{bmatrix}0  & 0 \\	1 & 0 \end{bmatrix}$                \\
			& &&\\\hline
		\end{tabular}
		\scriptsize
		{\raggedright \begin{center}Here, $x \in \{0,1\}$ represented with row-wise elements and  $y \in \{0,1\}$  are indexed column-wise in the strategy distribution $P(x,y|  a,b,w=1,2)$, where $a,b \in \{0,1\}$\end{center} \par}
	
	\end{table}

For the case of $W=3$ in this section, $P(X,Y|A,B,W=3)$ can be created as follows,
\begin{itemize}
\item For the pair $(a,b) \in \{(0,0), (0,1), (1,0)\} $, the term $ A.B  = 0 $, so, for these three cases, $Y = X$ as per \eqref{eq: y_one_way}.
Thus for these three cases,
\begin{align}\label{eq:cc1}
P(x,y|a,b,w=3) & = \I{x=y}P(x|a,b,w=3)=\I{x=y}P(x|w=3) \nonumber\\
&= \frac{1}{2}\begin{bmatrix}
	1  & 0 \\
	0 & 1
\end{bmatrix}.
\end{align}
\item For the pair $(a,b) = (1,1)$, the term $ A.B  = 1 $, so $Y \neq X$.
\begin{align}\label{eq:cc2}
P(x,y|a,b,w=3) & = \I{x \neq y}P(x|a,b,w=3)=\I{x \neq y}P(x|w=3) \nonumber\\
&= \frac{1}{2}\begin{bmatrix}
	0  & 1 \\
	1 & 0
\end{bmatrix}.
\end{align}
\end{itemize}

Using \eqref{eq:com}, we want to verify that the $\strat$ matches Example \eqref{Table : ccex2}. Consider $(a,b) = (0,0)$, and using the matrix representation again,
\begin{align}
	P(x,y|  a =b=0) & = \sum_{w \in \{1,2,3\}} P(x,y |  a=b=0,w) P(w) \nonumber  \\
	&= \frac{1}{3} \left( \begin{bmatrix}0  & 0 \\	1 & 0 \end{bmatrix}+ \begin{bmatrix}1  & 0 \\	0 & 0 \end{bmatrix} + \frac{1}{2}\begin{bmatrix}
		1  & 0 \\
		0 & 1
	\end{bmatrix} \right) \nonumber\\
	& =  \begin{bmatrix} \frac{1}{2} & 0 \\	\frac{1}{3} & \frac{1}{6} \end{bmatrix},
\end{align}
which is same as top-left matrix of Example \eqref{Table : ccex2}. Next, consider $(a,b) = (1,1)$, as below,
\begin{align}
	P(x,y|  a =b=0)
	&= \frac{1}{3} \left( \begin{bmatrix}0  & 1 \\	0 & 0 \end{bmatrix}+ \begin{bmatrix}0  & 0 \\	1 & 0 \end{bmatrix} + \frac{1}{2}\begin{bmatrix}
		0  & 1 \\
		1 & 0
	\end{bmatrix} \right) \nonumber\\
	& =  \begin{bmatrix} 0&\frac{1}{2} \\  \frac{1}{2}&0 \end{bmatrix}.
\end{align}
Similarly, it can be verified for the remaining $(a,b)$ pairs, using \eqref{eq:com}, along with Table \eqref{Table:2w3} and  \eqref{eq:cc1}.

We can once again reflect on the communication involved in this example. \A{} and \B{} were provided access to only passive common randomness by \Wit{}, which we noted was insufficient to generate the strategy $\strat$. However they can combine this passive common randomness with explicit, direct communication to generate this strategy. It is known in games that equilibria of noncooperative games that are equipped with a noisy communication system can be simulated by a correlated strategy involving one way communication with a mediator, and vice-versa~\cite{myerson1997game}. Something analogous to this seems to have played out here: instead of communicating confidentially with \Wit{}, the players have used the passive common randomness from \Wit{} as noise and coupled that with direct communication to produce the same strategy as they did with confidential communication.

This interchangeability of direct noisy communication and confidential mediated communication poses a dilemma for ascertaining if an information structure is respected. Even if one takes the view that confidential mediated communication is permissible under a decentralized information structure, the above example shows that such indirect communication is statistically indistinguishable from certain forms of direct communication (both the underlying mechanisms produce the same conditional distribution or strategy). Thus from the strategy alone one cannot say if players have in fact communicated. This further underlines the need to have a careful (possibly sample-path based)  definition of the information structure.

\section{No-signaling and posterior distributions} \label{sec:ns}
We now consider the no-signaling conditions, \ie, \eqref{eq:ns}  again in an attempt to interpret them in a somewhat different way.

\begin{proposition}
A strategy $\strat$ satisfies the no-signaling conditions
\eqref{eq:ns}  if and only if following posterior condition is satisfied
\begin{equation}
	P(A|  B,Y) = P(A |  B) , \ \ P(B |  A, X) = P(B |  A)   \label{eq:posterior}
\end{equation}
\end{proposition}

\begin{proof}
Consider forward direction, \ie, no-signaling condition implies posterior,
\begin{equation*}
	P(A|B,Y) = \frac{P(Y|A,B) P(A,B)}{P(B,Y)} = \frac{P(Y|B) \  P(A|B)}{P(Y|B)} = P(A|B)
\end{equation*}
Here, we used the no-signaling condition in obtaining the second equation. Similarly it can be shown for $P(B |  A, X) = P(B |  A)$.

Conversely, the condition \eqref{eq:posterior} on the posteriors implies the no-signaling condition,
\begin{equation*}
	P(X|A,B) = \frac{P(B|A,X) P(A,X)}{P(A,B)} = \frac{P(B|A)}{P(B|A)} P(X|A) = P(X|A)
\end{equation*}
Similarly it can be shown for $P(Y|A,B) = P(Y|B)$.
\end{proof}

We know that the no-signaling condition \eqref{eq:ns} says that action $X$ given $A$ is chosen without reference to observation $B$, and similarly for $Y$. The posterior condition, \eqref{eq:posterior} says that the belief of \B{} about \A{}'s observation $A$ is independent of \B{}'s action $Y$ given his observation $B$. In other words, \B{}'s action has no additional ``information'' to provide about \A{}'s observation than \B{}'s observation already provides.

Notice in \eqref{eq: y_one_way} in Section \ref{sec:one-way}, the action $Y$ is created using both $A,B$ and also action $X$. Despite this the  overall distribution $\strat$ follows the no-signaling condition and also the posterior condition \eqref{eq:posterior}. In this case, \B{} knows the realization of observation $A$ of \A{} and uses it in constructing $Y$, and yet, posterior condition \eqref{eq:posterior} holds.

\section{Discussion}
It has been shown through simple counterexample to Conjecture \ref{conj:conj1}  that our current interpretation of the no-signaling condition insofar as describing information structure in decentralized control, is incomplete. From the above examples we have constructed, we have seen that the no-signaling condition allows strategies where direct and indirect communication are performed. A statistical condition such as no-signaling seems to us to be insufficient to capture the presence or absence of communication, and to thereby serve as a boundary on permissible strategies under a decentralized information structure.

In summary, these counterexamples challenge our assumptions about the scope and limitations of no-signaling.
They call for deeper investigations into the nature of no-signaling and a more careful definition of information structure in decentralized control.
The question ``what is the set of all strategies that adhere to the information structure'' still remains unanswered.
Alongside, another question remains unanswered: what is the meaning of no-signaling if it is not ``no-communication''?

\bibliography{ref.bib}

\begin{thebibliography}{10}
\expandafter\ifx\csname url\endcsname\relax
  \def\url#1{\texttt{#1}}\fi
\expandafter\ifx\csname urlprefix\endcsname\relax\def\urlprefix{URL }\fi
\expandafter\ifx\csname href\endcsname\relax
  \def\href#1#2{#2} \def\path#1{#1}\fi

\bibitem{ananthram2007commonrandom}
V.~Ananthram, V.~Borkar, Common randomness and distributed control: A
  counterexample, Systems and Control Letters (2007).
\newblock \href
  {https://doi.org/https://doi.org/10.1016/j.sysconle.2007.03.010}
  {\path{doi:https://doi.org/10.1016/j.sysconle.2007.03.010}}.

\bibitem{deshpande2022quantum}
S.~A. Deshpande, A.~A. Kulkarni, The quantum advantage in decentralized
  control, Mathematical Control and Related Fields, under review (2023).
\newblock \href {https://doi.org/10.48550/ARXIV.2207.12075}
  {\path{doi:10.48550/ARXIV.2207.12075}}.

\bibitem{deshpande2022binary1}
S.~A. Deshpande, A.~A. Kulkarni, Quantum advantage in binary teams and the
  coordination dilemma: Part {I}, arXiv:2307.01762, under review by IEEE
  Transactions on Control of Networked Systems (2023).

\bibitem{deshpande2022binary2}
S.~A. Deshpande, A.~A. Kulkarni, Quantum advantage in binary teams and the
  coordination dilemma: Part {II}, arXiv:2307.01766, under review by IEEE
  Transactions on Control of Networked Systems (2023).

\bibitem{ho1980team}
Y.~Ho, Team decision theory and information structures, Proceedings of the IEEE
  68~(6) (1980) 644--654.

\bibitem{barrett2005nonlocal}
J.~Barrett, N.~Linden, S.~Massar, S.~Pironio, S.~Popescu, D.~Roberts,
  \href{https://link.aps.org/doi/10.1103/PhysRevA.71.022101}{Nonlocal
  correlations as an information-theoretic resource}, Phys. Rev. A 71 (2005)
  022101.
\newblock \href {https://doi.org/10.1103/PhysRevA.71.022101}
  {\path{doi:10.1103/PhysRevA.71.022101}}.
\newline\urlprefix\url{https://link.aps.org/doi/10.1103/PhysRevA.71.022101}

\bibitem{cover2012elements}
T.~M. Cover, J.~A. Thomas, Elements of information theory, John Wiley \& Sons,
  2012.

\bibitem{saldi2022geometry}
N.~Saldi, S.~Y{\"u}ksel, \href{https://doi.org/10.1214/20-PS356}{{Geometry of
  information structures, strategic measures and associated stochastic control
  topologies}}, Probability Surveys 19~(none) (2022) 450 -- 532.
\newblock \href {https://doi.org/10.1214/20-PS356}
  {\path{doi:10.1214/20-PS356}}.
\newline\urlprefix\url{https://doi.org/10.1214/20-PS356}

\bibitem{laloe2012we}
F.~Lalo{\"e}, Do we really understand quantum mechanics?, Cambridge University
  Press, 2012.

\bibitem{popescu1994nonlocquantum}
D.~R. Sandu~Popescu, Quantum nonloc ality as an axiom, Foundations of Physics
  24 (1994) 379--385.

\bibitem{myerson1997game}
R.~B. Myerson, Game theory: analysis of conflict, Harvard university press,
  1997.

\bibitem{clauser1969chsh}
J.~F. Clauser, M.~A. Horne, A.~Shimony, R.~A. Holt,
  \href{https://link.aps.org/doi/10.1103/PhysRevLett.23.880}{Proposed
  experiment to test local hidden-variable theories}, Phys. Rev. Lett. 23
  (1969) 880--884.
\newblock \href {https://doi.org/10.1103/PhysRevLett.23.880}
  {\path{doi:10.1103/PhysRevLett.23.880}}.
\newline\urlprefix\url{https://link.aps.org/doi/10.1103/PhysRevLett.23.880}

\end{thebibliography}

\section{Appendix}\label{sec: appendix}

We now show that the distribution $P(X,Y |  A,B)$ in Example \ref{Table : ccex2} is another counterexample of  Conjecture\ref{conj:conj0}. We wish to prove the claim : `the distribution $P(X,Y |  A,B)$ in Example \ref{Table : ccex2} lies in $\NS \backslash \Lrs$.' The proof of this claim uses the CHSH inequality \cite{clauser1969chsh}, that is described in short below.

Consider two agents \A{} and \B{}, with observations $a,b \in \{0,1\}$ and actions $x,y \in \{0,1\}$. Then by $\langle(-1)^x (-1)^y\rangle_{a,b}$ denote,
\begin{equation}\label{eq:corr}
	\langle(-1)^x (-1)^y\rangle_{a,b} = \sum_{y \in \{0,1\}}\sum_{x \in \{0,1\}} P(x,y|a,b) (-1)^x (-1)^y.
	\end{equation}
The CHSH inequality is given by,
\begin{align}
	\langle(-1)^x (-1)^y\rangle_{0,0} &  + \langle(-1)^x (-1)^y\rangle_{0,1} + \nonumber\\
	 \langle(-1)^x (-1)^y\rangle_{1,0} -& \langle(-1)^x (-1)^y\rangle_{1,1} \ \ \leq 2 \label{eq:CHSH}
	\end{align}

\begin{lemma}
    	If a distribution $\strat$ lies in $\Lrs$, then it must satisfy the CHSH inequality \eqref{eq:CHSH}.
\end{lemma}
		\begin{proof}
			Suppose there exists a passive common randomness $W$, such that $P(X,Y|A,B,W)$ satisfies \eqref{eq:nseq}. Then the correlation \eqref{eq:corr} can be rewritten as,
			\begin{align*}
				\langle(-1)^x &(-1)^y\rangle_{a,b} = \sum_{y \in \{0,1\}}\sum_{x \in \{0,1\}} \sum_{w} P(x,y|a,b,w) P(w) (-1)^x (-1)^y \\
				& = \sum_{w}  P(w)\sum_{y \in \{0,1\}}(-1)^y P(y|a,b,w) \sum_{x \in \{0,1\}} P(x|a,b,w) (-1)^x
				\end{align*}
The LHS of \eqref{eq:CHSH} simplifies as
				\begin{align}
					&\sum_{w} P(w) \sum_{y =0}^{1}(-1)^y  \sum_{x =0}^1 (-1)^x \nonumber\\
					&\quad( P(y|0,0,w) P(x|0,0,w) + P(y|1,0,w) P(x|1,0,w)+ \nonumber\\
					 & \quad P(y|0,1,w) P(x|0,1,w)- P(y|1,1,w) P(x|1,1,w) ) \label{eq:xCHSH}\\
					 &=\sum_{w} P(w) \sum_{x,y =0}^{1}(-1)^{x+y} ( \ P(x|0,0,w)\left(P(y|0,0,w) + P(y|1,1,w)\right) + \nonumber\\
					 \ \ & P(x|1,1,w) \left(P(y|0,0,w) - P(y|1,1,w) \right) \ ) \nonumber\\
					& =  \sum_{w} P(w) \sum_{x,y =0}^{1} (-1)^{x+y}	\left( P(x,y|0,0,w) +  P(x,y|1,1,w) \right) \nonumber\\
					& \leq  \sum_{w} P(w) \sum_{x,y =0}^{1}  	\left( P(x,y|0,0,w) +  P(x,y|1,1,w) \right) = 2. \nonumber
 					\end{align}
 					Here, \eqref{eq:xCHSH} is simplified using the property that $X\indep B|A,W$ and $Y \indep A|B,W$. This gives us the CHSH inequality\eqref{eq:CHSH}. Note, this proposition and proof is a manifestation of classical CHSH property as originally studied in \cite{clauser1969chsh}.
			\end{proof}
Now coming back to the claim about our Example \ref{Table : ccex2}.

\begin{lemma}
    The distribution $P(X,Y |  A,B)$ in Example \ref{Table : ccex2} lies in $\NS \backslash \Lrs$.
\end{lemma}
\begin{proof}
Using \eqref{eq:ex-ns}-\eqref{eq:ex-ns2}, we know that $P(x|a,b) = P(x|a) \ \forall b$ and $P(y|a,b) = P(y|b) \  \forall a$, so we know that this Example \eqref{Table : ccex2} lies in $\NS$.
It suffices to show that $\strat$ violates the condition \eqref{eq:CHSH}. Consider the correlation for $(a,b) = (0,0)$,
\begin{align*}
	\langle(-1)^x (-1)^y\rangle_{0,0} &= \sum_{y \in \{0,1\}}\sum_{x \in \{0,1\}} P(x,y|a=0,b=0) (-1)^x (-1)^y \\
	&= \frac{1}{2} (1) + \frac{1}{6}(1) + \frac{1}{3} (-1) = \frac{1}{3}.
	\end{align*}
Similarly for $(a,b) = (1,0)$, the correlation is  $ \langle(-1)^x (-1)^y\rangle_{1,0} = \frac{1}{3}$.
For $(a,b) = (1,1)$,
\begin{align*}
	\langle(-1)^x (-1)^y\rangle_{1,1} &= \sum_{y \in \{0,1\}}\sum_{x \in \{0,1\}} P(x,y|a=1,b=1) (-1)^x (-1)^y \\
	&= \frac{1}{2} (-1) + \frac{1}{2} (-1)= -1.
\end{align*}
For $(a,b) = (0,1)$,
\begin{align*}
	\langle(-1)^x (-1)^y\rangle_{0,1} &= \sum_{y \in \{0,1\}}\sum_{x \in \{0,1\}} P(x,y|a=0,b=1) (-1)^x (-1)^y \\
	&= \frac{1}{2} (1) + \frac{1}{2} (1)= 1.
\end{align*}
So, the LHS of \eqref{eq:CHSH} is
\begin{equation}
	\frac{1}{3} + 1 + \frac{1}{3} - (-1) = 2 + \frac{2}{3} > 2,
	\end{equation}
    showing that $\strat $ in  Example \ref{Table : ccex2} lies in $ \NS\backslash\Lrs.$
	\end{proof}
\end{document}